%% file: ms.tex
\pdfoutput=1
\documentclass{article}
\usepackage[utf8]{inputenc}
\usepackage{geometry}
\usepackage{liyang}
\include{vlecomte}

\newcommand{\Sc}{{\bar{S}}}
\newcommand{\xSc}{{x_\Sc}}
\newcommand{\fS}{{f_{S|\xSc}}}
\newcommand{\maxd}{{Cw \log (3/\epsilon)}}
\newcommand{\uastOne}{{100Cw\log (k+2) \log (3/\epsilon)}}
\newcommand{\uastTwo}{{100\maxd}}
\DeclareMathOperator{\Enc}{Encode}
\DeclareMathOperator{\Dec}{Decode}
\DeclareMathOperator{\cover}{cover}
\DeclareMathOperator{\numCovers}{numCovers}

\renewcommand{\emph}[1]{{\sl #1}}

\begin{document}\title{Sharper bounds on the Fourier concentration of DNFs
 \vspace{15pt}}

\author{Victor Lecomte\vspace{8pt} \\ \hspace{0pt}{\sl Stanford }\vspace{8pt} \\ vlecomte@stanford.edu \and 
\and Li-Yang Tan \vspace{8pt} \\ \hspace{-8pt} {\sl Stanford}\vspace{8pt} \\ liyang@cs.stanford.edu}

\date{\vspace{15pt}\small{\today}}

\maketitle

\begin{abstract} 
In 1992 Mansour proved that every size-$s$ DNF formula is Fourier-concentrated on $s^{O(\log\log s)}$ coefficients.   We improve this to $s^{O(\log\log k)}$ where $k$ is the \emph{read number} of the DNF.  Since~$k$ is always at most~$s$, our bound matches Mansour's for all DNFs and strengthens it for small-read ones.  The previous best bound for read-$k$ DNFs was $s^{O(k^{3/2})}$.   For $k$ up to \modif{$\tilde{\Omega}(\log\log s)$}{$\tilde{\Theta}(\log\log s)$}, we further improve our bound to the optimal $\poly(s)$; previously no such bound was known for any $k = \omega_s(1)$.

Our techniques involve new connections between the term structure of a DNF, viewed as a set system, and its Fourier spectrum. 

\end{abstract} 

\thispagestyle{empty}

\newpage
\setcounter{page}{1}

\section{Introduction} 

The relationships between combinatorial and analytic measures of Boolean function complexity is the subject of much study.  A classic result of this flavor is Mansour's theorem~\cite{Man92}, which shows that every size-$s$ DNF formula is Fourier-concentrated on $s^{O(\log\log s)}$  coefficients \modif{}{(that is, it is well-approximated by a polynomial with $s^{O(\log \log s)}$ monomials)}.  More precisely: \bigskip

\noindent{{\bf Mansour's theorem.}} {\it For every size-$s$ DNF $f$ and  every $\eps $, the Fourier spectrum of $f$ is $\eps$-concentrated on $(s/\eps)^{O(\log\log(s/\eps)\log(1/\eps))}$  coefficients.  }
\bigskip

\modif{}{However, Mansour conjectured that this bound was not tight, and that the correct bound was actually polynomial in $s$.}\bigskip

\noindent\modif{}{{{\bf Mansour's conjecture.}} {\it For every size-$s$ DNF $f$ and  every $\eps $, the Fourier spectrum of $f$ is $\eps$-concentrated on $s^{O_\eps(1)}$  coefficients.  }}
\bigskip

Our main result is a sharpening of Mansour's theorem that takes the \emph{read number} of the DNF into account.  We say that a DNF is \emph{read-$k$} if every variable occurs in at most $k$ of its terms. 
\begin{theorem}
\label{thm:intro:improve} 
For every size-$s$ read-$k$ DNF $f$ and  every $\eps$, the Fourier spectrum of $f$ is $\eps$-concentrated on $(s/\eps)^{O(\log\log k \log(1/\eps))}$  coefficients. 
\end{theorem}

Since $k$ is always at most $s$, our bound matches Mansour's for all DNFs (indeed, it slightly improves the dependence on $\eps$) and strengthens it for small-read ones.  The dependence on $k$ in~\Cref{thm:intro:improve} is a doubly-exponential improvement of the previous best bound of $s^{O(k^{3/2})}$ for read-$k$ DNFs~\cite{ST19}; this was in turn an exponential improvement of an~$s^{O(16^k)}$ bound by Klivans, Lee, and Wan~\cite{KLW10}, who gave the first nontrivial bounds for $k\ge 2$. 

For small values of $k$, we further improve our bound to the optimal $\poly(s)$: 

\begin{theorem}
\label{thm:intro:small-read}  For every size-$s$ DNF $f$ with read up to \modif{$\tilde{\Omega}(\log\log s)$}{$\tilde{\Theta}(\log\log s)$} and for every $\eps$, the Fourier spectrum of $f$ is $\eps$-concentrated on $(s/\eps)^{O(\log(1/\eps))}$  coefficients. 
\end{theorem} 

Previously no $\poly(s)$ bound was known for any $k = \omega_s(1)$ (even for constant $\eps$).  

Regarding the dependence on $\eps$ in these bounds, Mansour showed a lower bound $s^{\tilde{\Omega}(\log(1/\eps))}$ on the sparsity of any polynomial that $\eps$-approximates {\sc Tribes}, a read-once DNF formula.

 \Cref{thm:intro:improve,thm:intro:small-read} immediately yield faster membership query algorithms for agnostically learning small-read DNF formulas under the uniform distribution.   This is via a powerful technique of Gopalan, Kalai, and Klivans~\cite{GKK08}, showing that if every function in a concept class over $\bits^n$ can be $\eps$-approximated by $t$-sparse polynomials, then it can be agnostically learned in time $\poly(n,t,1/\eps)$.   As this implication is blackbox and by now standard, we do not elaborate further.

\subsection{Other related work} 

Recent work of Kelman, Kindler, Lifshitz, Minzer, and Safra~\cite{KKLMS20} proves that every boolean function $f$ is Fourier concentrated on $\mathbb{I}(f)^{O_\eps(\mathbb{I}(f))}$  coefficients, where $\mathbb{I}(f)$ is the total influence of~$f$.  Since $\mathbb{I}(f) \le O(\log s)$ for size-$s$ DNFs $f$, this result also recovers Mansour's bound as a special case and strengthens it for DNFs with small total influence (modulo the dependence on $\eps$).  

This result is incomparable to~\Cref{thm:intro:improve,thm:intro:small-read}.  On one hand it is more general, applicable to all functions rather than just DNF formulas.  On the other hand, there are small-read DNFs that saturate the $\mathbb{I}(f) \le O(\log s)$ bound (e.g.~{\sc Tribes} is a read-once DNF with total influence $\Theta(\log s)$). 

\subsection{Our techniques}
\label{sec:techniques}
It is well known that small-size DNFs are well-approximated by small-width\footnote{\modif{}{the \emph{width} of a DNF is the maximal length number of variables queried in a single term}} DNFs (see \Cref{fact:small-size-small-depth} in the preliminaries), so in this discussion and most of the paper, we focus on the ``width'' version of the question: that is, showing Fourier concentration for a width-$w$ DNF $f$.

Our main conceptual contribution is to bound Fourier coefficients $\Fourier{f}(S)$ by a quantity that depends on how $S$ relates to the term structure of $f$ as a DNF
: we bound $\abs{\Fourier{f}(S)}$ by the probability over a random input $x$ that $S$ can be ``covered'' by the terms that $x$ satisfies\footnote{\modif{}{we say $x$ ``satisfies'' some term if the values given by $x$ make the term output true}} in $f$ (that is, the probability that each variable of $S$ occurs in some satisfied term).
Let us call this probability the \emph{cover probability} of $S$.
We use this bound on $\abs{\Fourier{f}(S)}$ twice, to prove the two main ingredients of our proof: a Fourier 1-norm bound (\Cref{lemma:onenorm-u}) and a 2-norm bound (\Cref{lemma:twonorm-u}).

The next three headings describe what happens in Sections~\ref{sec:one-norm} through~\ref{sec:using-read}.

\paragraph{The 1-norm bound.} The first ingredient is a sharpening of Mansour's \cite{Man92} bound on the Fourier 1-norm due to low-degree monomials.
The broad structure of the proof in \cite{Man92} is to first show that $f$ is concentrated up to degree $O(w)$, then to show that the Fourier 1-norm up to that degree is at most $w^{O(w)}$, which gives concentration on the same number of coefficients.
As Mansour himself showed, this bound on the 1-norm is tight, even for read-once DNFs like {\sc Tribes}. 
Therefore, $w^{\Theta(w)}$ seemed to be the end of the story for 1-norm-based methods.

It turns out that we can make Mansour's Fourier 1-norm bound more precise by splitting monomials $x^S$ into groups $\calS_u$ based on (roughly) the size $u$ of a minimal union of terms that includes $S$.
We show a bound of $\binom{u}{O(w)}$ on the Fourier 1-norm due to $\calS_u$.
Note that this is a strict improvement on Mansour's bound: the minimal cover of a set $S$ of size $O(w)$ can involve at most $|S|$ terms and therefore will have total size at most $w|S|=O(w^2)$, in which case our bound $\binom{u}{O(w)}$ matches Mansour's bound $w^{O(w)}$.
On the other hand, for $u \ll w^2$, our bound $\binom{u}{O(w)}$ is much smaller than $w^{O(w)}$.

\modif{We prove this bound using a specialized version of Håstad's switching lemma \cite{Hastad87} that takes this cover size $u$ into account during the encoding phase.
The reason we are able to take it into account is that, unlike in the usual application of Håstad's switching lemma, the set of free variables corresponds exactly to the sets $S$ whose Fourier coefficients we are trying to bound, instead of being random sets of $\Theta(n/w)$ variables.
This also means we do not need to use any facts about the way that random restrictions affect the Fourier spectrum.}{We prove this bound by tweaking Razborov's \cite{Razborov95} proof of Håstad's switching lemma \cite{Hastad87} to take this cover size $u$ into account during the encoding phase.
We first relate the absolute value of the Fourier coefficient $\abs{\Fourier{f}(S)}$ to the probability that a random restriction to $S$ has decision tree depth $|S|$, then use Razborov's encoding to show that this probability is small.
Then, instead of separately identifying each variable of $S$ by encoding their position within a term (which costs $\log w$ bits per variable), we encode their positions all together within the union of the terms they appear in (which costs $\binom{u}{|S|}$ in total).}

\paragraph{The 2-norm bound.} For the second ingredient, we give concrete bounds on the absolute value of the \emph{individual} Fourier coefficients (as opposed to bounding their sum).
Indeed, if the Fourier 1-norm due to a family of sets is $\leq M$ and, for each $S$ in that family, $|\Fourier{f}(S)| \leq \delta$, then the total Fourier weight due to that family is at most $M\delta$.
In particular, if $M\delta \ll 1$, then we can simply discard that family.
For $S \in \calS_u$, we can bound $|\Fourier{f}(S)|$ by (roughly) $2^{-u}$ times the number of ways to cover $S$ by terms of $f$.

\paragraph{Concluding using small read.} The read of $f$ allows us to bound the number of ways a set can be covered by terms. We do this in two regimes:
\begin{itemize}
    \item In general, if the read of $f$ is $k$, then it is easy to see that any set $S$ can only be (minimally) covered in $(k+1)^{|S|}$ ways.
    Thus we can bound the 2-norm due to family $\calS_u$ by $M\delta \leq \binom{u}{O(w)}2^{-u}(k+1)^{O(w)}$, which is negligibly small for $u=\omega(w \log k)$.
    This means we can cut off at $u=O(w \log k)$, getting 1-norm at most $\binom{O(w \log k)}{O(w)} = (\log k)^{O(w)}$ for the remaining coefficients, and thus concentration on $(\log k)^{O(w)}$ coefficients (\Cref{thm:improve}).
    \item If the read of $f$ is small enough ($k \leq \frac{\log w}{\log \log w}$), then we can improve on the trivial $(k+1)^{|S|}$ bound by using a combinatorial result of \cite{ST19} which states that the expected number of satisfied terms of an unbiased read-$k$ DNF is $O(k)$. This allows us to cut off as low as $u=O(w)$, giving concentration on $2^{O(w)}$ coefficients, and thus proving Mansour's conjecture for that entire family of functions (\Cref{thm:small-read}).
\end{itemize}
We note that this choice of coefficients (keeping only $\Fourier{f}(S)$ for sets $S$ that are contained in a small union of terms) follows almost exactly the approach suggested by Lovett and Zhang~\cite{LZ19} for proving Mansour's conjecture, although we did not end up using their sparsification result.


\section{Preliminaries}
In this section we define Boolean functions, the Fourier spectrum, and DNFs along with their complexity metrics (size, width, read).
We also recall some facts that are used in Mansour's original proof \cite{Man92}.
For an in-depth treatment, see \cite{ODBook,AoBF-YouTube}. 

\subsection{Boolean functions and Fourier analysis}
We view Boolean functions as functions $f:\InBool^n \to \R$, where an input of $-1$ represents ``true'' and an input of $1$ represents ``false''.
However, for output values, we will use $1$ for ``true'' and $-1$ for ``false'', as usual.

This choice of input values makes the Fourier spectrum of $f$ more convenient to define.
For $S \subseteq [n]$, let
\[
\Fourier{f}(S) = \E_{x \in \InBool^n}[f(x)x^S]
\]
where $x^S \coloneqq \prod_{i\in S}x_i$.
Any Boolean function is uniquely represented as a multilinear polynomial, where the coefficients are exactly the values $\Fourier{f}(S)$, which we call \emph{Fourier coefficients}:
\[
f(x) = \sum_{S \subseteq [n]} \Fourier{f}(S) x^S.
\]
We say $f$ is \emph{$\epsilon$-concentrated} on a family $\calS \subseteq 2^{[n]}$ if $\sum_{S \not\in \calS}\Fourier{f}(S)^2 \leq \epsilon$, and we say that $f$ is $\epsilon$-concentrated on $M$ coefficients if there is such an $\calS$ with $\abs{\calS} \leq M$.
We define the \emph{Fourier $p$-norm} of $f$ as
\[
\p*{\sum_{S \subseteq [n]} \abs{\Fourier{f}(S)}^p}^{1/p},
\]
and the special case $p=1$ has the following property:
\begin{fact}[Exercise 3.16 in \cite{ODBook}]
\label{fact:one-norm-implies-concentration}
Let $M = \sum_{S \subseteq [n]}\abs{\Fourier{f}(S)}$ be the Fourier 1-norm of $f$, then $f$ is $\epsilon$-concentrated on $M^2/\epsilon$ coefficients.
\end{fact}

Finally, we say a function $g:\InBool^n \to \OutBool$ \emph{$\epsilon$-approximates} a function $f:\InBool^n \to \OutBool$ if they differ on at most an $\epsilon$ fraction of inputs, that is,
\[
\Pr_{x\in\InBool^n} [f(x) \neq g(x)] \leq \epsilon.
\]
\begin{fact}[Exercise 3.17 in \cite{ODBook}]
If $g$ $\epsilon_1$-approximates $f$ and is $\epsilon_2$-concentrated on a family $\calS$, then $f$ is $2(\epsilon_1+\epsilon_2)$-concentrated on $\calS$.
\end{fact}

\subsection{DNFs}
A function $f:\InBool^n \to \OutBool$ is a DNF if it can be represented as an OR of ANDs of the input variables.
Each AND is called a \emph{term}, and the number of terms is called the \emph{size}.
We write a size-$s$ DNF as $f = T_1 \vee \cdots \vee T_s$, and by abuse of notation, we frequently use $T_j$ to represent the set of variables in the $j\th$ term.

A DNF has \emph{width} $w$ if each of its terms queries at most $w$ variables (i.e. $|T_j|\leq w$ for all $j$), and \emph{read} $k$ if each variable occurs in at most $k$ terms.
We only use $w$ and $k$ as upper bounds (except in \Cref{thm:exactly-w}, where it is explicitly stated), which justifies us occasionally assuming ``large enough $w$'' or ``large enough $k$''.
A small-size DNF can be approximated by a small-width DNF:
\begin{fact}
\label{fact:small-size-small-depth}
Let $f$ be a size-$s$ DNF.
Then there is a DNF $g$ of width $\log(s/\epsilon)$ that $\epsilon$-approximates $f$. In fact, $g$ is simply obtained from $f$ by dropping some terms from $f$, so $g$'s size and read are both at most $f$'s.
\end{fact}
This means that to prove \Cref{thm:intro:improve} and \Cref{thm:intro:small-read}, it is enough to prove the corresponding statements for width $w$: respectively, that width-$w$ read-$k$ DNFs are $\epsilon$-concentrated on $2^{O(w \log \log k \log 1/\epsilon)}$ coefficients (\Cref{thm:improve}), and that width-$w$ DNFs with read up to $\tilde{\Omega}(\log w)$ are $\epsilon$-concentrated on $2^{O(w \log 1/\epsilon)}$ coefficients (\Cref{thm:small-read}).

\subsection{Restrictions and Håstad's switching lemma}
Let $S \subseteq [n]$ be a set of variables, let $\Sc= [n]\setminus S$, and let $\xSc \in \InBool^{\Sc}$ be an assignment to only the variables in $\Sc$.
Then the \emph{restriction} of $f$ to $S$ at $\xSc$ is the function
$\fS : \InBool^S \to \R$
that maps $x_S \in \InBool^S$ to $f(x_S \circ \xSc)$, where $\circ$ denotes the act of combining vectors $x_S$ and $\xSc$ into a vector of $\InBool^{S \cup \Sc} = \InBool^n$.
We call the variables of $S$ ``free'' and the variables of $\Sc$ ``fixed''.

Let $\DT(f)$ denote $f$'s \emph{decision tree depth}: the smallest depth of a decision tree computing $f$ exactly.
Håstad's switching lemma \cite{Hastad87} states that a random restriction of $f$ (where both $S$ and $\xSc$ are chosen randomly
) is unlikely to have high decision tree depth.
It is used twice in the proof of Mansour's theorem \cite{Man92}: once to show that $f$ is concentrated on low-degree monomials (\Cref{fact:low-degree} below), and once to show that after removing the high-degree monomials, the Fourier 1-norm is low.

\begin{fact}
\label{fact:low-degree}
There is a constant $C>1$ such that any width-$w$ DNF $f$ is $\epsilon$-concentrated up to degree $Cw \log 1/\epsilon$.
\end{fact}

In the next section, we prove a variant of Håstad's switching lemma in order to improve on the second application.
Since the proof format of switching lemmas is quite unusual and complex, it is definitely helpful to be familiar with the proof of the original version beforehand.
An excellent pedagogical presentation of the proof can be found in \cite{switching-lemma-notes}.

\subsection{Miscellaneous}
We use $\log$ to denote the base-$2$ logarithm (though the base will rarely matter).
For a finite set $S$, we define $\binom{S}{k}$ to be the family of subsets of $S$ that have size $k$.
Finally, for a finite alphabet $A$, we use $A^\ast$ to denote the set of strings over $A$, and we write the empty string as $()$.


\section{Cover sizes and the switching lemma}
\label{sec:one-norm}
As we mentioned before, the proof of Mansour's theorem works by first proving concentration on degree up to $O(w \log 1/\epsilon)$, then showing that the Fourier 1-norm due to monomials of degree at most $d$ is at most $w^{O(d)}$.
In this paper, we use the first part as is, and focus on improving the second part: the $w^{O(d)}$ bound on the Fourier 1-norm for degree at most $d$.
The only fact we will ever use about the Fourier spectrum of $f$ is the following lemma, which bounds the absolute value $|\Fourier{f}(S)|$ of the Fourier coefficient of $S$ by the probability that a random restriction to $S$ has decision tree depth $|S|$.

\begin{lemma}
\label{lemma:evasive}
Let $S \subseteq [n]$. Then $|\Fourier{f}(S)| \leq \Pr_{\xSc \in \InBool^{\Sc}}[\DT(\fS) = |S|]$.
\end{lemma}
\begin{proof}
First, observe that for any Boolean function $g: \InBool^m \to \OutBool$, if $\DT(g) < m$, then $\Fourier{g}([m])=0$ (indeed, the degree of $g$ is at most $\DT(g)$), and in all cases $|\Fourier{g}([m])| \leq 1$.
Because of this, $\Fourier{g}([m]) \leq \One[\DT(g) = m]$, and applying this to $g\coloneqq\fS$, we get
\begin{equation*}
    |\Fourier{f}(S)|
    = \abs*{\E_{\xSc}[\Fourier{\fS}(S)]}
    \leq \E_{\xSc}\sqb*{\abs{\Fourier{\fS}(S)}}
    \leq \Pr_{\xSc}[\DT(\fS) = |S|].\qedhere
\end{equation*}
\end{proof}

\begin{remark}
For intuition, we think it is helpful to mentally replace the probability in \Cref{lemma:evasive} with the ``cover probability'' of $S$ which we mentioned in \Cref{sec:techniques}: the probability over a random input $x$ that every variable of $S$ is present in at least one term that $x$ satisfies.
More broadly, the notion of ``cover by terms'' will be a key player throughout this whole paper.
A formal link between these two probabilities is that when $\DT(\fS)=|S|$, it is possible to assign the variables of $S$ to make sure that they are all involved in at least one satisfied term.\footnote{If $f$ is monotone, then it is easy to see: the fact that $\DT(\fS)=|S|$ shows that all variables of $S$ are present in some term that is alive in $\fS$, and by assigning all variables of $S$ to true, we can satisfy all those terms. For $f$ non-monotone it is only slightly subtler.} This directly implies the following fact (although we will not use in this paper).
\end{remark}
\begin{fact}
Let $S \subseteq [n]$. Then $|\Fourier{f}(S)| \leq 2^{|S|}\Pr_{x \in \InBool^n}[\text{$S$ is covered by satisfied terms}]$.
\end{fact}

For pedagogical reasons, we first reprove Mansour's bound on the Fourier 1-norm at degree $d$ as a warmup, but using our \Cref{lemma:evasive}.
We will then slightly tweak the proof so that it tells us more about which coefficients contribute the most to the 1-norm (within a given degree $d$).

We start by summing up \Cref{lemma:evasive} over all sets $S$ of size $d$, and get a bound on the Fourier 1-norm at degree $d$ that depends on the number of restrictions of $f$ to $d$ variables that require full decision tree depth.
\begin{lemma}
\label{lemma:onenorm-d}
$\sum_{S:|S|=d}|\Fourier{f}(S)| \leq 2^{-(n-d)}\times\#\set{(S,\xSc):|S|=d \wedge \DT(\fS) = d}$.
\end{lemma}
\begin{proof}
By \Cref{lemma:evasive}, $\sum_{S:|S|=d}|\Fourier{f}(S)| \leq \sum_{S:|S|=d} \Pr_{\xSc \in \InBool^{\Sc}}[\DT(\fS) = d]$, and we can rewrite each probability in this sum as the average $2^{-(n-d)}\times\#\set{\xSc:\DT(\fS) = d}$.
\end{proof}

This now allows us to reprove Mansour's bound on the 1-norm, using Håstad's switching lemma  \cite{Hastad87}.
In the following proof, the sets of free variables correspond exactly to the sets whose Fourier coefficients we are trying to bound, rather than random selections of a $\Theta(1/w)$ fraction of the variables, as is usual.
This allows us to avoid using facts about how random restrictions affect the Fourier spectrum, and will later allow us to take into account properties of the specific sets $S$ we are encoding.
\begin{lemma}
\label{lemma:reprove-mansour}
$\sum_{S:|S|=d}|\Fourier{f}(S)| \leq \binom{wd}{d} 2^{2d} = w^{O(d)}$.
\end{lemma}
\begin{proof}
By \Cref{lemma:onenorm-d}, we only need to show
\[
\#\set{(S,\xSc):|S|=d \wedge \DT(\fS) = d} \leq 2^{n+d} \times \binom{wd}{d}.
\]
We will follow Razborov's \cite{Razborov95} version of the proof of the switching lemma.
This version uses an encoding argument, which consists of an encoding algorithm and a decoding algorithm.
The encoding algorithm will take as inputs a set $S$ of size $d$ and an assignment $\xSc$ of the remaining variables such that $\DT(\fS) = d$ (one of the couples $(S,\xSc)$ that we want to show are rare).
It will produce as outputs an assignment $x \in \InBool^n$ of \emph{all} the variables, an element $\sigma \in \binom{[wd]}{d}$ (indicating a $d$-size subset of $[wd]$), and a $d$-bit binary string $a \in \InBool^d$.
The decoding algorithm will uniquely recover $S$ and $\xSc$ from $(x,\sigma,a)$, showing that the encoding is injective, and therefore that there are only
\[\underbrace{2^n}_{\text{number of $x$'s}} \times \underbrace{\binom{wd}{d}}_{\text{number of $\sigma$'s}} \times \underbrace{2^d}_{\text{number of $a$'s}}\]
possible values of $(S,\xSc)$ such that $|S|=d$ and $\DT(\fS) = d$, as desired.

The encoding algorithm Algorithm~\ref{algorithm:encode} will identify a number of terms $T_{j_1}, ,\ldots, T_{j_l}$ which together contain every variable of $S$, and complete the partial assignment $\xSc$ in a way that makes those terms easy to find (with some ``hints'' stored in $a$). It then suffices to encode which variables of the union $T_{j_1} \cup \cdots \cup T_{j_l}$ belong to $S$, which is the role of the second output $\sigma$. \modif{}{Note that this is different from Razborov's encoding, instead of separately encoding the position of each variable within its term, we encode the positions of all of the variables at once, within the \emph{union} of all of the terms $T_{j_1}, \ldots, T_{j_l}$}

\modif{}{More precisely (but still in words), Algorithm~\ref{algorithm:encode} does the following:
\begin{itemize}
    \item Initialize $S'$ to $S$. $S'$ will represent the set of free variables.
    \item Initialize two partial assignments $x^{\mathrm{sat}}, x^{\mathrm{dt}}$ to $\xSc$. $x^{\mathrm{sat}}$ will be the string sent to the encoder, and is progressively made to satisfy the terms $T_{j_1}, \ldots, T_{j_l}$ (which will be identified later). $x^{\mathrm{dt}}$, on the other hand, will progressively assign the variables of $S$ in a way that maximizes the decision tree depth.
    \item Initialize a Boolean string $a$, which will contain the set of changes that the decoder needs to make to the variables of $S$ in order to go from $x^{\mathrm{sat}}$ to $x^{\mathrm{dt}}$ during the decoding.
    \item Initialize a string $c$, which will contain the union of the variables of $T_{j_1}, \ldots, T_{j_l}$, \emph{in the order that they appear} (i.e. $c$ starts with the variables of $T_{j_1}$, then the variables of $T_{j_2} \setminus T_{j_1}$, etc.).
    \item While $S'$ is not empty (i.e. there are free variables in $x^{\mathrm{sat}}$ and $x^{\mathrm{dt}}$):
    \begin{itemize}
        \item Let $T_j$ be the first term not fixed by $x^{\mathrm{dt}}$ (we will call $j_1, \ldots, j_l$ the successive values that $j$ takes in this loop).
        \item Let $S_j$ be the set of free variables in $T_j$.
        \item Let $x_{S_j}^{\mathrm{sat}}$ be the assignment to the variables of $S_j$ that make $T_j$ satisfied.
        \item Let $x_{S_j}^{\mathrm{dt}}$ be the assignment to the variables of $S_j$ that maximizes the remaining decision tree depth of $f$ (after it is restricted by $x^{\mathrm{sat}} \circ x_{S_j}^{\mathrm{sat}}$).
        \item Extend $x^{\mathrm{sat}}$ with $x_{S_j}^{\mathrm{sat}}$ and $x^{\mathrm{dt}}$ with $x_{S_j}^{\mathrm{dt}}$.
        \item Remove the variables of $S_j$ from $S'$.
        \item Add $x_{S_j}^{\mathrm{dt}}$ to string $a$.
        \item Add to $c$ the variables of $T_j$ that were not contained in the previously identified terms.
    \end{itemize}
    \item Let $\sigma$ be the set of indices within $c$ where the variables of $S$ are located ($\sigma$ is a subset of $[wd]$).
    \item Return $x^{\mathrm{sat}}$, $\sigma$, and $a$.
\end{itemize}
}

\modif{}{The decoding algorithm will progressively re-identify terms $T_{j_1}, \ldots, T_{j_l}$ by looking at the first satisfied term in $f$, and progressively replacing the assignments $x_{S_j}^{\mathrm{sat}}$ (which satisfy $T_j$) by the assignments $x_{S_j}^{\mathrm{dt}}$ (which maximize decision tree depth).
It will identify which variables of $T_j$ are part of $S$ by using $\sigma$.
Crucially, string $c$ in the encoding algorithm, which represents the union $T_{j_1} \cup \cdots \cup T_{j_l}$, follows the order of the terms, so even though at the $r\th$ step the decoding algorithm knows only terms $T_{j_1}, \ldots, T_{j_r}$, it can reconstruct the first $\abs{T_{j_1} \cup \cdots \cup T_{j_r}}$ characters of $c$, and therefore correctly recover the set $S_r$ using $\sigma$.}

\modif{}{More precisely (but still in words), Algorithm~\ref{algorithm:decode} does the following:
\begin{itemize}
    \item Initialize the set of variables $S$ to an empty set.
    \item Initialize string $c$ to an empty string. This string replicates string $c$ from the decoding algorithm, and will take the exact same sequence of values as the algorithm progresses.
    \item While $|S|<d$ (we have not found all the variables yet):
    \begin{itemize}
        \item Let $T_j$ be the first term satisfied by $x$.
        \item Add to $c$ all the variables of $T_j$ that have not been added to it in previous iterations.
        \item Within the variables newly added to $c$, look at the one whose indices within $c$ are in $\sigma$, and call this set of variables $S_j$. This will be the same as set $S_j$ in the encoding algorithm.
        \item Replace $x$'s assignment of the variables of $S_j$ by the values $a_{|S|+1}, \ldots, a_{|S|+|S_j|}$. This replaces $x_{S_j}^{\mathrm{sat}}$ by $x_{S_j}^{\mathrm{dt}}$ within $x$ (where $x_{S_j}^{\mathrm{sat}}$ and $x_{S_j}^{\mathrm{dt}}$ are defined in the encoding algorithm).
        \item Extend $S$ with $S_j$.
    \end{itemize}
    \item Return $S$ and the values of $x$ on $\Sc$.
\end{itemize}
}

\begin{algorithm}[H]
\label{algorithm:encode}
\caption{$\Enc(S,\xSc)$}
$S' \gets S$\;
$x^{\mathrm{sat}}, x^{\mathrm{dt}} \gets \xSc$\;
$a \gets () \in \InBool^\ast$\tcp*{empty Boolean string}
$c \gets () \in [n]^\ast$\tcp*{empty string of variables}
\While{$S' \neq \emptyset$}{
    $j \gets \min\set{j : T_j(x^{\mathrm{dt}}) \not\equiv 0}$\tcp*{$T_j$ is the first term unfixed by $x^{\mathrm{dt}}$}
    $S_j \gets T_j \cap S'$\tcp*{the set of variables unfixed in $T_j$}
    $x_{S_j}^{\mathrm{sat}} \gets$ the assignment of $S_j$ such that $T_j(x_{S_j}^{\mathrm{dt}} \circ x_{S_j}^{\mathrm{sat}}) \equiv 1$\;
    $x_{S_j}^{\mathrm{dt}} \gets$ any assignment of $S_j$ such that $\DT(f_{(S'\setminus S_j) | (x^{\mathrm{dt}}\circ x_{S_j}^{\mathrm{dt}})}) = |S' \setminus S_j|$\;
    $x^{\mathrm{sat}} \gets x^{\mathrm{sat}} \circ x_{S_j}^{\mathrm{sat}}$\;
    $x^{\mathrm{dt}} \gets x^{\mathrm{dt}} \circ x_{S_j}^{\mathrm{dt}}$\;
    $S' \gets S' \setminus S_j = S' \setminus T_j$\;
    Append $a$ with $x_{S_j}^{\mathrm{dt}}$\;
    Append $c$ with all variables of $T_j$ that are not yet in $c$\;
}
$\sigma \gets \set{k \in |c| : c_k \in S}$\tcp*{the positions of the variables of $S$ within $c$}
\KwRet{$(x^{\mathrm{sat}},\sigma,a)$}
\end{algorithm}

\begin{definition}
\label{def:jr}
Let $j_1 < \cdots < j_l$ be the successive values taken by $j$ in Algorithm~\ref{algorithm:encode}.
\end{definition}

\begin{claim}
\label{claim:encode-correct}
The encoding algorithm runs successfully, and in particular,
\begin{enumerate}[(i)]
    \item at the start of each run of the while loop, $x^{\mathrm{sat}}$ and $x^{\mathrm{dt}}$ are both assignments of all variables except $S'$, and $\DT(f_{S'|x^{\mathrm{dt}}}) = |S'|$;
    \item $x_{S_j}^{\mathrm{dt}}$ always exists;
    \item $S = S_{j_1} \cup \cdots \cup S_{j_l}$;
    \item $c$ contains all variables of $S$;
    \item $|c| \leq wd$.
\end{enumerate}
\end{claim}
\begin{proof}
\begin{enumerate}[(i)]
    \item Clear by induction and the choice of $x_{S_j}^{\mathrm{dt}}$.
    \item By (i), $\DT(f_{S'|x^{\mathrm{dt}}}) = |S'|$, and there is always a way to assign $|S_j|$ variables without decreasing the decision tree depth by more than $|S_j|$.
    \item Clear since $S_{j_r} \subseteq S$ for all $r \in [l]$ and by the end of the algorithm, $\emptyset = S' = S \setminus S_{j_1} \setminus \cdots \setminus S_{j_l}$.
    \item By (iii), $S = S_{j_1} \cup \cdots \cup S_{j_l} \subseteq T_{j_1} \cup \cdots \cup T_{j_l}$, and $c$ is constructed to contain exactly the variables of $T_{j_1} \cup \cdots \cup T_{j_l}$.
    \item Each $T_{j_r}$ contains at least one variable of $S$, so $l < |S| = d$, and each term has at most $w$ variables, so $|c| = |T_{j_1} \cup \cdots \cup T_{j_l}| \leq wd$.\qedhere
\end{enumerate}
\end{proof}

\begin{algorithm}[H]
\label{algorithm:decode}
\caption{$\Dec(x,\sigma,a)$}
$S \gets \emptyset$\;
$c \gets () \in [n]^\ast$\;
\While{$|S| < d$}{
    $j \gets \min\set{j : T_j(x) = 1}$\tcp*{$T_j$ is the first term satisfied by $x$}
    Append $c$ with all variables of $T_j$ that are not yet in $c$\;
    $S_j \gets \set{c_k : k \in \sigma \wedge k \leq |c|} \setminus S$\;
    Replace $x$'s assignment of $S_j$ by the values $a_{|S|+1}, \ldots, a_{|S|+|S_j|}$\;
    $S \gets S \cup S_j$\;
}
\KwRet{$(S,x|_{\Sc})$}
\end{algorithm}

\begin{claim}
The successive values that $j$ and $S_j$ take in $\Dec(\Enc(S, \xSc))$ are exactly $j_1, \ldots, j_l$ and $S_{j_1}, \ldots, S_{j_l}$ (i.e. the same values they took in $\Enc(S, \xSc)$).
\end{claim}
\begin{proof}
We will show by induction that when the $r\th$ run of the while loop starts, we have
\begin{itemize}
    \item $S = S_{j_1} \cup \cdots S_{j_{r-1}}$;
    \item $c$ contains the variables of the union $T_{j_1} \cup \cdots \cup T_{j_{r-1}}$;
    \item $x=\xSc \circ x_{S_{j_1}}^\text{dt} \circ \cdots \circ x_{S_{j_{r-1}}}^\text{dt} \circ x_{S_{j_{r}}}^\text{sat} \circ \cdots \circ x_{S_{j_{l}}}^\text{sat}$.
\end{itemize}
If this is the case $r$, then we argue that the first term satisfied by $x$ at the start of the $r\th$ run is $T_{j_r}$. Indeed,
\begin{itemize}
    \item $x' \coloneqq \xSc \circ x_{S_{j_1}}^\text{dt} \circ \cdots \circ x_{S_{j_{r-1}}}^\text{dt}$ is exactly the value of $x^\text{dt}$ in the $r\th$ run of the while loop in $\Enc(S, \xSc)$;
    \item by \Cref{claim:encode-correct}(ii), $f$ is undecided by $x'$, so $f_{(S_{j_r} \cup \cdots \cup S_{j_l})|x'}$ has no satisfied term, and in particular, by the definition of $T_{j_r}$, all terms before $T_{j_r}$ are \emph{un}satisfied by $x'$;
    \item $x_{S_{j_r}}^\text{sat}$ is defined to satisfy $T_{j_r}$.
\end{itemize}
Therefore, $j$ will be assigned to $j_r$, $c$ will be appropriately extended, which will allow the algorithm to correctly recover $S_{j_r}$ add it to $S$, and replace the values of $x$ on $S_{j_r}$ by the values $x_{S_{j_r}}^\text{dt}$ stored in $a$, proving the inductive hypothesis for $r+1$.
\end{proof}

As a consequence, the first output of $\Dec$ is indeed $S_{j_1} \cup \cdots \cup S_{j_r} = S$, and the second output must be $\xSc$, since the values of the variables outside of $S$ were never modified by either algorithm.
Therefore, for all valid $(S,\xSc)$, we have $\Dec(\Enc(S, \xSc)) = (S,\xSc)$, which concludes the proof of \Cref{lemma:reprove-mansour}.
\end{proof}

The main cost in the above lemma is the $\binom{wd}{d}$ factor, where the $wd$ comes from the fact that covering a set $S$ of size $d$ by terms of $f$ can take up to $d$ terms, for a total size $wd$.
This suggests that one might get savings if we know that a set $S$ is ``typically'' covered by terms whose union has size much less than $wd$.
This motivates the following definitions.

\begin{definition}[$\cover(S,\xSc)$, $u(S,\xSc)$]
Given $(S,\xSc)$ such that $\DT(\fS) = |S|$, let $\cover(S, \xSc) \coloneqq \set{j_1, \cdots, j_l}$ and $u(S, \xSc) \coloneqq |T_{j_1} \cup \cdots \cup T_{j_l}|$ where $j_1, \ldots, j_l$ are the successive values of $j$ obtained when running $\Enc(S,\xSc)$ (just like in \Cref{def:jr}).
\end{definition}
\begin{fact}
\label{fact:cover-covers}
For any $(S,\xSc)$ such that  $\DT(\fS) = |S|$,
\begin{enumerate}[(i)]
    \item $S \subseteq \bigcup_{j \in \cover(S,\xSc)} T_j$;
    \item $\abs{\cover(S,\xSc)} \leq |S|$;
    \item for all $j \in \cover(S,\xSc)$, $T_j$ is alive under the partial assignment $\xSc$ (not yet fixed to $0$ or $1$);
    \item for all $j \in \cover(S,\xSc)$, $T_j$ contains at least one variable of $S$.
\end{enumerate}
\end{fact}
\begin{proof}
Clear by inspecting Algorithm~\ref{algorithm:encode}.
\end{proof}
\begin{definition}[$\calS_{d,u}$]
Let $\calS_{d,u}$ be the family of sets $S$ of $d$ variables such that the most frequent value of $\cover(S, \xSc)$ is $u$ (breaking ties arbitrarily). In other words,
\[
\calS_{d,u} = \set{S \subseteq [n]: |S| = d \wedge u = \argmax_u \Pr_{\xSc}[\DT(\fS) = d \wedge u(S,\xSc) = u]}.
\]
\end{definition}

When $S \in \calS_{d,u}$, we get to assume that $\cover(S,\xSc) = u$ for only a small extra factor.
In doing this, we replace Lemmas~\ref{lemma:evasive} and~\ref{lemma:onenorm-d} by the following two lemmas.
\begin{lemma}
\label{lemma:abs-fourier-u}
If $S \in \calS_{d,u}$, then
$\abs{\Fourier{f}(S)} \leq (wd+1) \times \Pr_\xSc[\DT(\fS) = d \wedge u(S,\xSc) = u]$.
\end{lemma}
\begin{proof}
As we showed before, the maximum value of $u(S,\xSc)$ is $wd$, so it can take at most $wd+1$ values.\footnote{In fact at most $wd-d+1$ values, since $u(S,\xSc)$ is always at least $d$.}
Therefore, since $u$ maximizes $\Pr_\xSc[\DT(\fS) = d \wedge u(S,\xSc) = u]$, we have
\[
\Pr_\xSc[\DT(\fS) = d] \leq (wd+1) \Pr_\xSc[\DT(\fS) = d \wedge u(S,\xSc) = u].
\]
We then conclude by \Cref{lemma:evasive}.
\end{proof}

\begin{lemma}
$\sum_{S \in \calS_{d,u}} \abs{\Fourier{f}(S)} \leq (wd+1)2^{-(n-d)}\times\#\set{(S,\xSc):|S|=d \wedge \DT(\fS) = d \wedge u(S,\xSc) = u}$
\end{lemma}
\begin{proof}
Sum up \Cref{lemma:abs-fourier-u} then transform the probability into an average, similar to \Cref{lemma:onenorm-d}.
\end{proof}

We can now run the same encoding argument, but with the guarantee that $|T_{j_1} \cup \ldots\cup T_{j_l}|=u$, to show the following more specialized 1-norm bound.
\begin{lemma}
\label{lemma:onenorm-u}
$\sum_{S \in \calS_{d,u}} \abs{\Fourier{f}(S)} \leq (wd+1)\binom{u}{d} 2^{2d} = \p*{\frac{u}{d}}^{O(d)}$.
\end{lemma}
\begin{proof}
Same as the proof of \Cref{lemma:reprove-mansour}, but now with a bound of $u$ on the final size of $c$, which allows us to pick $\sigma$ from the smaller set $\binom{[u]}{d}$.
\end{proof}


\section{In how many ways can a set of variables be covered by terms?}
\label{sec:two-norm}
At this point, a fair question would be: what was the point of proving these refined bounds for the Fourier 1-norm based on the typical cover size $u$ if the worst case ($d=w$, $u=wd$) gives $w^{\Theta(w)}$ anyway?
To see how this is useful, let us look at the contribution of family $\calS_{d,u}$ to the Fourier weight:
\begin{equation}
\label{eq:max-fourier}
\sum_{\calS_{d,u}} \Fourier{f}(S)^2
\leq \p*{\sum_{\calS_{d,u}} |\Fourier{f}(S)|} \times \max_{\calS_{d,u}} |\Fourier{f}(S)|
\leq (wd+1)\binom{u}{d}2^{2d} \max_{\calS_{d,u}} |\Fourier{f}(S)|.
\end{equation}
If we can bound $\max_{\calS_{d,u}} |\Fourier{f}(S)|$ by something that decreases faster than $(wd+1)\binom{u}{d}2^{2d}$ increases, then we can bound $\calS_{d,u}$'s contribution to the Fourier weight.

Now, how can we use the fact that $S \in \calS_{d,u}$ to bound $|\Fourier{f}(S)|$?
Well, we know from \Cref{lemma:abs-fourier-u} that
\begin{equation}
\label{eq:recap-lemma}
\abs{\Fourier{f}(S)} \leq (wd+1) \times \Pr_\xSc[\DT(\fS) = d \wedge u(S,\xSc) = u].
\end{equation}
It turns out that we can bound the above probability by bounding the number of ways that one can cover $S$ by a union of terms whose total size is $u$.
\begin{definition}
Given $S \in \calS_{d,u}$, let
\[
\numCovers(S) \coloneqq \#\set*{\calT \subseteq [m] : S \subseteq \bigcup_{j \in \calT} T_j \wedge \abs{\calT} \leq |S| \wedge  \abs*{\bigcup_{j \in \calT} T_j} = u}.
\]
This roughly represents the ``number of ways $S$ can be covered by terms'', and as we will see, this is an upper bound on the number of possible values of $\cover(S,\xSc)$.
\end{definition}

\begin{lemma}
\label{lemma:abs-leq-num-covers}
Let $S \in \calS_{d,u}$. Then
\begin{align*}
\Pr_\xSc[\DT(\fS) = d \wedge u(S,\xSc) = u]
&\leq 2^{-(u-d)} \#\set{\cover(S,\xSc):\DT(S,\xSc) = d \wedge u(S,\xSc) = u}\\
&\leq 2^{-(u-d)} \numCovers(S).
\end{align*}
\end{lemma}
\begin{proof}
Let us first show the first inequality.
Let $\{j_1, \ldots, j_l\}$ be the value of $\cover(S,\xSc)$ for some $\xSc$ such that $\DT(\fS)=d$ and $u(S,\xSc) = u$.
Then by \Cref{fact:cover-covers}(iii), terms $T_{j_1}$ must all be alive under partial assignment $\xSc$.
This constraint fixes the values of all variables in $(T_{j_1} \cup \cdots \cup T_{j_l}) \setminus S$, of which there are $u-d$, so this value of $\cover(S,\xSc)$ can only contribute $2^{-(u-d)}$ to the probability.
This shows the first inequality. The second inequality is a consequence of \Cref{fact:cover-covers}(i) and (ii).
\end{proof}

Putting together \eqref{eq:max-fourier}, \eqref{eq:recap-lemma} and \Cref{lemma:abs-leq-num-covers}, we obtain the following.
\begin{lemma}
\label{lemma:twonorm-u}
$\sum_{\calS_{d,u}} \Fourier{f}(S)^2 \leq (wd+1)^2\binom{u}{d}2^{3d}2^{-u} \times\max_{|S|=d}\numCovers(S)$.
\end{lemma}
Among the factors in front of the $\max$, $2^{-u}$ is the one that will dominate.
So if we can prove that $\max_{|S|=d}\numCovers(S) \ll 2^u$, that is, that
there are significantly less than $2^u$ ways to cover a set of $d$ variables by a union of terms of total size $u$, then we can show that $\calS_{d,u}$'s contribution to the Fourier weight is negligible.


\section{Proof of our main theorems}
\label{sec:using-read}
Let us summarize the approach. We approximate $f$ in three ways:
\begin{enumerate}
    \item First, \Cref{fact:low-degree} tells us that $f$ is $(\epsilon/3)$-concentrated on degree at most $\maxd$ (for some constant $C>1$). In particular, this shows that $f$ is $(\epsilon/3)$-concentrated on the union of the families $\calS_{d,u}$ for $d\leq \maxd$.
    \item Second, using \Cref{lemma:twonorm-u} we will show that among those families $\calS_{d,u}$, $f$ is $(\epsilon/3)$-concentrated on the families with $u\leq u^\ast$ for some $u^\ast$ (for \Cref{thm:improve}, we obtain $u^\ast=O(w \log k \log 1/\epsilon)$, and for \Cref{thm:small-read}, we obtain $u^\ast=O(w \log 1/\epsilon)$). In other words, we will prove
    \[
        \sum_{u = \floor{u^\ast}+1}^\infty\sum_{d =0}^{\floor{\maxd}} \sum_{S \in \calS_{d,u}} \Fourier{f}(S)^2 \leq \epsilon/3.
    \]
    \item Finally, using \Cref{lemma:onenorm-u}, we will show that the Fourier 1-norm for $u \leq u^\ast$ is at most some quantity $M$. In other words, we will prove
    \[
        \sum_{u=0}^{\floor{u^\ast}}\sum_{d =0}^{\floor{\maxd}} \sum_{S \in \calS_{d,u}} \abs{\Fourier{f}(S)} \leq M.
    \]
    By \Cref{fact:one-norm-implies-concentration}, this implies that the sum of the corresponding monomials is $(\epsilon/3)$-concentrated on $3M^2/\epsilon$ coefficients.
\end{enumerate}
Those three approximations together will show that the original function $f$ is $(\epsilon/3+\epsilon/3+\epsilon/3=\epsilon)$-concentrated on $3M^2/\epsilon$ coefficients.

To make our job in step 2 slightly easier in advance of proving \Cref{thm:improve} and \Cref{thm:small-read}, let us sum up, specialize and simplify \Cref{lemma:twonorm-u}.
\begin{corollary}
\label{cor:twonorm-u}
For large enough $w$, if $u \geq 100\maxd$, then
\[
    \sum_{d =0}^{\floor{\maxd}} \sum_{S \in \calS_{d,u}} \Fourier{f}(S)^2 \leq 2^{-u/2} \max_{|S|\leq \maxd}\numCovers(S).
\]
\end{corollary}
\begin{proof}
By \Cref{lemma:twonorm-u} and ugly arithmetics, for large enough $w$,
\begin{align*}
    \sum_{d =0}^{\floor{\maxd}} \sum_{S \in \calS_{d,u}} \Fourier{f}(S)^2
    &\leq \sum_{d =0}^{\floor{\maxd}} (wd+1)^2\binom{u}{d}2^{3d}2^{-u} \times\max_{|S|=d}\numCovers(S)\\
    &\leq \sum_{d =0}^{\floor{\maxd}} (Cw^2\log (3/\epsilon)) +1)^2\binom{u}{\maxd}2^{3\maxd}2^{-u}\\
    &\qquad\qquad\times\max_{|S|\leq \maxd}\numCovers(S)\\
    &\leq u^{O(1)} 2^{(\log(\frac{eu}{\maxd})+3)\maxd} 2^{-u} \times\max_{|S|\leq \maxd}\numCovers(S)\\
    &\leq 2^{-u/2} \max_{|S|\leq \maxd}\numCovers(S).\qedhere
\end{align*}
\end{proof}

\subsection{General improvement to Mansour's theorem}
How large can $\numCovers(S)$ get for $|S|=d$ if $f$ has read $k$?
In other words, how many ways are there to cover a set $S$ by terms of a read-$k$ DNF?
By \Cref{fact:cover-covers}(iv), each term in the cover must contain a variable of $S$, and each variable is present in at most $k$ terms, so there are at most $kd$ terms to choose from.
In addition, by \Cref{fact:cover-covers}(ii), the cover can contain at most $d$ terms, so
\[
\numCovers(S) \leq \sum_{i=0}^d \binom{kd}{i} \leq \binom{kd+d}{d} \leq (e(k+1))^d = O(k)^d.
\]

Now, looking at \Cref{cor:twonorm-u}, we see that we need to choose $u^\ast$ big enough that this is much smaller than $2^{u^\ast/2}$. Thus it suffices to pick $u^\ast$ to be about $\Theta(d \log k) = \Theta(Cw\log k \log 1/\epsilon)$.
The following lemma makes this precise
\begin{lemma}
\label{lemma:concentration-dlogk}
For $u^\ast = \uastOne$,
\[
\sum_{u=\floor{u^\ast}+1}^{\infty} \sum_{d=0}^{\floor{\maxd}} \sum_{S \in \calS_{d,u}} \Fourier{f}(S)^2
\leq \epsilon/3.
\]
\end{lemma}
\begin{proof}
By \Cref{cor:twonorm-u},
\begin{align*}
    \sum_{u=\floor{u^\ast}+1}^{\infty} \sum_{d=0}^{\floor{\maxd}} \sum_{S \in \calS_{d,u}} \Fourier{f}(S)^2
    &\leq \sum_{u=\floor{u^\ast}+1}^{\infty} 2^{-u/2} \max_{|S| \leq \maxd}\numCovers(S)\\
    &\leq \sum_{u=\floor{u^\ast}+1}^{\infty} 2^{-u/2} (e(k+1))^\maxd\\
    &\leq \sum_{u=\floor{u^\ast}+1}^{\infty} 2^{-u/4} 2^{-25Cw\log (k+2) \log (3/\epsilon)} 2^{\log(e(k+1))\maxd}\\
    &\leq \sum_{u=\floor{u^\ast}+1}^{\infty} 2^{-u/4}\\
    &\leq \frac{1}{1-2^{-1/4}} \times 2^{-25Cw \log(k+2) \log (3/\epsilon)}\\
    &\leq \epsilon/3.\qedhere
\end{align*}
\end{proof}

Now, all we have to do is to plug this value of $u^\ast$ into \Cref{lemma:onenorm-u} to get the following theorem.
\begin{theorem}[``width'' version of \Cref{thm:intro:improve}]
\label{thm:improve}
Let $f$ be a width-$w$, read-$k$ DNF.
Then $f$ is $\epsilon$-concentrated on $\log (k+2)^{O(w \log 1/\epsilon)}$ coefficients.
\end{theorem}
\begin{proof}
By \Cref{fact:low-degree}, $f$ is $\epsilon/3$-concentrated up to degree $\maxd$, and by \Cref{lemma:concentration-dlogk}, the coefficients in $\calS_{d,u}$ for $d\leq \maxd$ and $u > u^\ast = \uastOne$ also only amount to Fourier weight at most $\epsilon/3$.
In addition, by \Cref{lemma:onenorm-u} the remaining coefficients have total 1-norm at most
\begin{align*}
    \sum_{u=0}^{\floor{u^\ast}} \sum_{d=0}^{\floor{\maxd}} \sum_{S \in \calS_{d,u}}\abs{\Fourier{f}(S)}
    &\leq \sum_{u=0}^{\floor{u^\ast}} \sum_{d=0}^{\floor{\maxd}} (wd+1)\binom{u}{d} 2^{2d}\\
    &\leq (u^\ast+1) \sum_{d=0}^{\floor{\maxd}} (wd+1)\binom{u^\ast}{d} 2^{2d}\\
    &\leq (u^\ast+1) (\maxd+1)(Cw^2\log (3/\epsilon) +1)\\
    & \ \ \ \  \times \binom{\uastOne}{\maxd} 2^{2\maxd}\\
    &\leq (u^\ast+1)(w \log 1/\epsilon)^{O(1)} \p{400e \log (k+2)}^{\maxd}\\
    &= \log (k+2)^{O(w \log 1/\epsilon)}.
\end{align*}
Therefore, using \Cref{fact:one-norm-implies-concentration} with error $\epsilon/3$, $f$ is $\epsilon$-concentrated on $3\p*{\log (k+2)^{O(w \log 1/\epsilon)}}^2/\epsilon = \log (k+2)^{O(w \log 1/\epsilon)}$ coefficients.
\end{proof}

\subsection{A proof of Mansour's conjecture for small enough read}
In the previous subsection, we bounded the number of covers of $S$ by $O(k)^d$, which was only small enough when $O(k)^d < 2^{u/4} \Leftrightarrow u = O(d \log k)$.
If we want to prove Mansour's conjecture, we need to do better: we need to bound the number of covers by $2^{u/4}$ for any $u = \omega(d)$.

The way to achieve this is to bound the number of terms that can be involved in the cover.
In the previous subsection, we simply observed that the cover is made of at most $|S|=d$ terms among the $kd$ terms that contain variables of $S$.
But when the read is small, we can do better.

To build some intuition, suppose that every term involves \emph{exactly} $w$ variables, rather than at most $w$.
Since every variable can only occur in at most $k$ terms, this would mean (by double counting) that any union of $l$ terms has total size at least $lw/k$.
Therefore, if we want the union to have size $u$, there can only be $ku/w$ terms in it.
Thus there would be at most
\[
    \binom{kd}{ku/w} \leq \binom{ku}{ku/w} \leq (ew)^{ku/w} = 2^{ku\log(ew)/w}
\]
ways to cover $S$ with total size $u$.
As long as $k \leq \frac{w}{4\log(ew)}$, this is at most $2^{u/4}$, and we easily get the following theorem.
\begin{theorem}
\label{thm:exactly-w}
Let $f$ be a DNF whose terms are all conjunctions of exactly $w$ variables, and suppose $f$ has read $k \leq \frac{w}{4\log(ew)}$. Then $f$ is $\epsilon$-concentrated on $2^{O(w \log 1/\epsilon)}$ coefficients.
\end{theorem}
\begin{proof}
By \Cref{fact:low-degree}, $f$ is $\epsilon/3$-concentrated up to degree $\maxd$.
Let $u^\ast = 100\maxd$. Then, by \Cref{cor:twonorm-u},
\begin{align*}
    \sum_{u=\floor{u^\ast}+1}^\infty \sum_{d=0}^{\floor{\maxd}} \sum_{S \in \calS_{d,u}} \Fourier{f}(S)^2
    &\leq \sum_{u=\floor{u^\ast}+1}^\infty 2^{-u/2} \max_{|S|\leq\maxd}\numCovers(S)\\
    &\leq \sum_{u=\floor{u^\ast}+1}^\infty 2^{-u/2} \binom{k\maxd}{ku/w}\\
    &\leq \sum_{u=\floor{u^\ast}+1}^\infty 2^{-u/2} \binom{ku}{ku/w}\\
    &\leq \sum_{u=\floor{u^\ast}+1}^\infty 2^{-u/2} 2^{ku\log(ew)/w}\\
    &\leq \sum_{u=\floor{u^\ast}+1}^\infty 2^{-u/4}\\
    &\leq \frac{1}{1-2^{-1/4}} \times 2^{-25Cw \log 1/\epsilon}\\
    &\leq \epsilon/3.
\end{align*}
In addition, by \Cref{lemma:onenorm-u}, the remaining coefficients have total 1-norm at most
\begin{align*}
    \sum_{u=0}^{\floor{u^\ast}} \sum_{d=0}^{\floor{\maxd}} \sum_{S \in \calS_{d,u}}\abs{\Fourier{f}(S)}
    &\leq \sum_{u=0}^{\floor{u^\ast}} \sum_{d=0}^{\floor{\maxd}} (wd+1)\binom{u}{d} 2^{2d}\\
    &\leq (u^\ast+1) (\maxd+1)(Cw^2\log (3/\epsilon) +1) 2^{u^\ast} 2^{2\maxd}\\
    &= 2^{O(w \log 1/\epsilon)}.
\end{align*}
Therefore, using \Cref{fact:one-norm-implies-concentration} with error $\epsilon/3$, $f$ is $\epsilon$-concentrated of $3\p*{2^{O(w \log 1/\epsilon)}}^2/\epsilon = 2^{O(w \log 1/\epsilon)}$ coefficients.
\end{proof}

However, this reasoning breaks down if some terms of $f$ are allowed to have width smaller than $w$.
For example, if $f$ contained the one-variable term $x_i$ for each $i \in S$, then the union could after all contain as many as $d$ terms, rather than $ku/w$.

We get out of this issue (though not without some loss) by using the following lemma from \cite{ST19}, which tells us in essence that $f$
cannot contain too many short terms without being very biased.
\begin{fact}[Lemma 1.1 in \cite{ST19}, rephrased]
\label{fact:st}
Let $f = T_1 \vee \cdots \vee T_s$ be a read-$k$ DNF. Then
\[
\sum_{j=1}^s 2^{-|T_j|} \leq k \ln \p*{\frac{1}{1-\Pr_x[f(x)]}}.
\]
\end{fact}
For the purposes of $\epsilon$-approximation, we can thus assume that
\[
\sum_{j=1}^s 2^{-|T_j|} \leq k \ln 1/\epsilon
\]
(otherwise, we can simply approximate $f$ by the constant $1$ function).

We can now show that the union will be made of few terms by proving the following combinatorial lemma.
\begin{lemma}
\label{lemma:combinatorial}
Let $A_1, \ldots, A_l$ be a family of finite sets such that
\begin{enumerate}[(i)]
    \item $|A_1| + \cdots + |A_l| \leq v$;
    \item $2^{-|A_1|} + \cdots + 2^{-|A_l|} \leq F$,
\end{enumerate}
with $v>F$.
Then $l \leq \frac{4v}{\log(v/F)}$.
\end{lemma}
\begin{proof}
Intuitively, the two constraints are in direct tension: if we want to keep $|A_r|$ small, this makes $2^{-|A_r|}$ big, and vice versa. 
So we will show that each set $A_r$ uses up a $\frac{\log(v/F)}{2v}$ fraction of the ``budget'' for either sum (i) or sum (ii).
Concretely, for any $A_r$, either $|A_r| \geq \log(v/F)/2 \geq \p*{\frac{\log(v/F)}{2v}}v$, or
\[
    2^{-|A_r|}
    \geq \sqrt{\frac{F}{v}}
    \geq \frac{F}{v}\times \frac{\log(v/F)}{2}
    = \p*{\frac{\log(v/F)}{2v}}F,
\]
where the second inequality comes from the fact that $\sqrt{x} \geq \frac{\log x}{2}$ for $x>0$, applied to $x=v/F$.
Both of those cases can only happen $\frac{2v}{\log(v/F)}$ times without violating either (i) or (ii), which means there can only be $l \leq \frac{4v}{\log(v/F)}$ sets in the family.
\end{proof}

Since $f$ is read-$k$, if a union of terms has size $u$, then the sum of the size of its terms is at most $ku$.
In addition, by \Cref{fact:st}, the terms $T_{j_1}, \ldots, T_{j_l}$ forming the union must obey
\[
2^{-|T_{j_1}|} + \cdots + 2^{-|T_{j_l}|} \leq \sum_{j=1}^s 2^{-|T_j|} \leq k\ln 1/\eps.
\]
Therefore, we can apply \Cref{lemma:combinatorial} with $v=ku$ and $F=k\ln 1/\epsilon$ to show that there are at most
\[
\binom{kd}{\frac{4ku}{\log(u/\ln(1/\epsilon))}}
\leq \binom{ku}{\frac{4ku}{\log(u/\ln(1/\epsilon))}}
\leq \p*{\frac{e\log(u/\ln(1/\epsilon))}{4}}^{\frac{4ku}{\log(u/\ln(1/\epsilon))}}
\leq 2^{ku\frac{4\log\log(u/\ln(1/\epsilon)))}{\log(u/\ln(1/\epsilon))}}
\]
ways to cover $S$ by a union of terms of total size $u$.

As long as $k \leq \frac{\log(u/\ln(1/\epsilon))}{16\log \log (u/\ln(1/\epsilon))}$ for the smallest value of $u$ we have to consider, this is at most $2^{u/4}$.
We will set $u^\ast \coloneqq \uastTwo$, so the smallest value we have to consider is
\[
\floor{u^\ast}+1 \geq \uastTwo \geq w \ln(1/\epsilon).
\]
Therefore, $k \leq \frac{\log w}{16\log \log w}$ suffices, and we get the following theorem.
\begin{theorem}[``width'' version of \Cref{thm:intro:small-read}]
\label{thm:small-read}
Let $f$ be a width-$w$ DNF that has read $k \leq \frac{\log w}{16\log \log w}$. Then $f$ is $\epsilon$-concentrated on $2^{O(w \log 1/\epsilon)}$ coefficients.
\end{theorem}
\begin{proof}
The proof is very similar to the proof of \Cref{thm:exactly-w}.
By \Cref{fact:low-degree}, $f$ is $\epsilon/3$-concentrated up to degree $\maxd$.
Let $u^\ast \coloneqq \uastTwo$. Then, by \Cref{cor:twonorm-u},
\begin{align*}
    \sum_{u=\floor{u^\ast}+1}^\infty \sum_{d=0}^{\floor{\maxd}} \sum_{S \in \calS_{d,u}} \Fourier{f}(S)^2
    &\leq \sum_{u=\floor{u^\ast}+1}^\infty 2^{-u/2} \max_{|S|\leq\maxd}\numCovers(S)\\
    &\leq \sum_{u=\floor{u^\ast}+1}^\infty 2^{-u/2} \binom{k\maxd}{\frac{4ku}{\log(u/\ln(1/\epsilon))}}\\
    &\leq \sum_{u=\floor{u^\ast}+1}^\infty 2^{-u/2} \binom{ku}{\frac{4ku}{\log(u/\ln(1/\epsilon))}}\\
    &\leq \sum_{u=\floor{u^\ast}+1}^\infty 2^{-u/2} 2^{ku\frac{4\log\log(u/\ln(1/\epsilon)))}{\log(u/\ln(1/\epsilon))}}\\
    &\leq \sum_{u=\floor{u^\ast}+1}^\infty 2^{-u/2} 2^{ku\frac{4\log\log(u^\ast/\ln(1/\epsilon)))}{\log(u^\ast/\ln(1/\epsilon))}}\\
    &\leq \sum_{u=\floor{u^\ast}+1}^\infty 2^{-u/4}\\
    &\leq \frac{1}{1-2^{-1/4}} \times 2^{-25Cw \log (3/\epsilon)}\\
    &\leq \epsilon/3.
\end{align*}
In addition, by \Cref{lemma:onenorm-u}, the remaining coefficients have total 1-norm at most
\begin{align*}
    \sum_{u=0}^{\floor{u^\ast}} \sum_{d=0}^{\floor{\maxd}} \sum_{S \in \calS_{d,u}}\abs{\Fourier{f}(S)}
    &\leq \sum_{u=0}^{\floor{u^\ast}} \sum_{d=0}^{\floor{\maxd}} (wd+1)\binom{u}{d} 2^{2d}\\
    &\leq (u^\ast+1) (\maxd+1)(Cw^2\log (3/\epsilon) +1) 2^{u^\ast} 2^{2\maxd}\\
    &= 2^{O(w \log 1/\epsilon)}.
\end{align*}
Therefore, using \Cref{fact:one-norm-implies-concentration} with error $\epsilon/3$, $f$ is $\epsilon$-concentrated of $3\p*{2^{O(w \log 1/\epsilon)}}^2/\epsilon = 2^{O(w \log 1/\epsilon)}$ coefficients.
\end{proof}


\section{Conclusion}
In this section, we present some open problems, and a direction that the results in this paper suggest.

\subsection{Open problems}
Besides the obvious open problem which is to prove Mansour's conjecture, we see two ways one could extend the results in our paper.

The first would be to improve the dependence on $k$ in~\Cref{thm:intro:improve}. After all, if three exponential improvements were possible starting from~\cite{KLW10}, how about a fourth? In fact, any significant improvement over the current dependence on $k$ would strictly improve Mansour's theorem even for general DNFs. Indeed, given that $k \leq s$, improving from $s^{O(\log \log k)}$ to, say, $s^{O(\log \log \log k)}$ would improve Mansour's theorem from $s^{O(\log \log s)}$ to $s^{O(\log \log \log s)}$.

The second (and perhaps easier) option would be to prove Mansour's conjecture for a bigger range of reads, improving on~\Cref{thm:intro:small-read}. Indeed, in~\Cref{thm:exactly-w}, we showed that Mansour's conjecture holds for $k$ up to $\Omega(w / \log w)$ instead of $\Omega(\log w / \log\log w)$ if all terms have exactly $w$ variables, instead of just at most $w$. To us, it intuitively feels like width exactly $w$ is the ``hardest case'', and it is hard to see how having shorter terms should not help a DNF have a much more spread-out Fourier spectrum, but we have not been able to make this intuition formal. In addition, our argument in~\Cref{thm:intro:small-read} does not feel tight: the way we use \cite{ST19}'s lemma (\Cref{fact:st}) feels ``wasteful'', since we apply it to only the very few terms that are involved in covering some set $S$, rather than to the entire DNF. Because of this, we conjecture that \Cref{thm:intro:small-read} can be improved with similar techniques to handle reads up to $\tilde{\Omega}(w) = \tilde{\Omega}(\log s)$ rather than the current $\tilde{\Omega}(\log \log s)$.

\subsection{A structure vs pseudorandomness approach to Mansour's conjecture?}
A recent trend in solving hard combinatorics problems has been the ``structure vs pseudorandomness'' paradigm, which consists in decomposing an object into a ``structured'' part and a ``pseudorandom'' part, where ``pseudorandom'' can mean stand for property that a randomly drawn object would typically have.
In particular, this paradigm has recently been used by Alweiss, Lovett, Wu and Zhang \cite{ALWZ20} to improve bounds on the sunflower lemma from $w^{O(w)}$ to $(\log w)^{O(w)}$.

We think that a similar argument can be applied to Mansour's conjecture.
In fact, our techniques (and in particular, \Cref{lemma:onenorm-u}) suggest a natural candidate for what it means to a DNF to be ``pseudorandom'': $f$ is pseudorandom if for any set of variables $S$, there are few ways to cover $S$ minimally using terms of $f$.\footnote{perhaps weighting covers with a factor $2^{-u}$ where $u$ is the size of the cover's union}
Indeed, random DNFs where $\poly(n)$ terms of size $\Theta(\log n)$ are drawn at random have this property (which gives an alternate proof of \cite{KLW10}'s results about random DNFs).

However, some work remains to be done in order to deal with the ``structured'' case.
Perhaps the most difficult case that remains unsolved is the random DNF where $n$ terms of size $\log n$ are drawn at random \emph{from the first $\log^2n$ variables only}.
Indeed, in this case, the overlaps between terms are so strong that \cite{KLW10}'s techniques become applicable, and each set $S \subseteq [\log^2n]$ of variables has many minimal covers by terms.
In fact, we personally know researchers who have devoted significant amounts of time attempting to either prove Mansour's conjecture for this DNF or use it as a counterexample.
Therefore, we feel that this DNF is a natural next challenge to attack, and we feel optimistic that if someone manages to prove Mansour's conjecture for it, they would be very close to proving the general case.

\section*{Acknowledgments}

We thank the anonymous reviewers, whose comments have helped improve this paper.  Li-Yang is supported by NSF CAREER Award 1942123. 
\newpage 
\bibliography{most-influential}{}
\bibliographystyle{alpha}

\end{document}

%% file: vlecomte.tex
\usepackage{amsmath,amsthm,amsfonts,amssymb,bm,bbm,booktabs,cleveref,enumerate,esvect,fancyhdr,framed,ifthen,listings,mathtools,soul,tabu,tikz}
\usepackage[scaled=0.8]{DejaVuSansMono}
\usepackage[normalem]{ulem}

\usepackage[ruled,lined]{algorithm2e}
\DontPrintSemicolon

\usetikzlibrary{positioning}
\tikzset{every picture/.style={semithick, font=\small}}

\renewcommand{\th}{^\text{th}}

\newcommand{\modif}[2]{#2}

\DeclarePairedDelimiter{\floor}{\lfloor}{\rfloor}

\DeclarePairedDelimiter{\set}{\{}{\}}
\DeclarePairedDelimiter{\p}{(}{)}
\DeclarePairedDelimiter{\sqb}{[}{]}
\DeclarePairedDelimiter{\abs}{|}{|}

\DeclareMathOperator{\One}{\mathbf{1}}

\DeclareMathOperator{\argmax}{arg\,max}

\newcommand{\InBool}{\set{-1,1}}
\newcommand{\OutBool}{\set{0,1}}
\newcommand{\Fourier}[1]{\widehat{#1}}

\DeclareMathOperator{\DT}{DT}

\newcommand{\calS}{\mathcal{S}}
\newcommand{\calT}{\mathcal{T}}